\documentclass[a4paper,11pt]{article}

\usepackage{hyperref}
\usepackage{a4wide}
\usepackage{authblk}
\usepackage{bm}   
\usepackage{amssymb,amsmath,amsthm} 
\usepackage[mathscr]{euscript}
\usepackage{stmaryrd}
\usepackage{tikz-cd}

\newtheorem{theorem}{Theorem}
\newtheorem{proposition}[theorem]{Proposition}

\newcommand{\As}{\mathscr{A}}
\newcommand{\Bs}{\mathscr{B}}

\newcommand{\sg}{\sigma}
\newcommand{\RA}{R^{\As}}
\newcommand{\RB}{R^{\Bs}}
\newcommand{\RaA}{R_{\alpha}^{\As}}
\newcommand{\RaB}{R_{\alpha}^{\Bs}}
\newcommand{\IMP}{\; \Rightarrow \;}
\newcommand{\AND}{\, \wedge \,}
\newcommand{\CS}{\mathcal{R}(\sg)}
\newcommand{\CSp}{\mathcal{R}_{\star}(\sg)}

\newcommand{\preford}{\sqsubseteq}

\newcommand{\IFF}{\Longleftrightarrow}
\newcommand{\rarr}{\rightarrow}
\newcommand{\id}{\mathsf{id}}
\newcommand{\ie}{\textit{i.e.}~}
\newcommand{\Count}{\#}
\newcommand{\LL}{\mathcal{L}}

\newcommand{\Lk}{\mathcal{L}_k}
\newcommand{\Lc}{\mathcal{L}(\Count)}
\newcommand{\Lck}{\mathcal{L}_{k}(\Count)}
\newcommand{\ELk}{\exists\mathcal{L}_{k}}
\newcommand{\Lvk}{\mathcal{L}^{k}}
\newcommand{\Lvck}{\mathcal{L}^{k}(\Count)}
\newcommand{\ELvk}{\exists\mathcal{L}^{k}}
\newcommand{\eqLk}{\equiv^{\Lk}}
\newcommand{\eqELk}{\equiv^{\ELk}}
\newcommand{\eqLck}{\equiv^{\Lck}}
\newcommand{\eqLvk}{\equiv^{\Lvk}}
\newcommand{\eqELvk}{\equiv^{\ELvk}}
\newcommand{\eqLvck}{\equiv^{\Lvck}}
\newcommand{\KK}{\mathcal{K}}
\newcommand{\eqL}{\equiv^{\LL}}
\newcommand{\vphi}{\varphi}
\newcommand{\iffdef}{\;\; \stackrel{\Delta}{\IFF} \;\;}
\newcommand{\GG}{\mathsf{G}}
\newcommand{\Ck}{\mathbb{C}_{k}}
\newcommand{\Cl}{\mathbb{C}_{l}}
\newcommand{\Kl}{\mathsf{Kl}}
\newcommand{\epsA}{\varepsilon_{\As}}
\newcommand{\eqak}{\rightleftarrows_{k}}
\newcommand{\eqbk}{\leftrightarrow_{k}}
\newcommand{\eqck}{\cong_{\Kl(\Ck)}}
\newcommand{\eqaCk}{\rightleftarrows_{k}^{\mathbb{C}}}
\newcommand{\eqbCk}{\leftrightarrow_{k}^{\mathbb{C}}}
\newcommand{\eqcCk}{\cong_{k}^{\mathbb{C}}}
\newcommand{\Alk}{A^{\leq k}}
\newcommand{\Blk}{B^{\leq k}}

\newcommand{\CC}{\mathcal{C}}
\newcommand{\Gd}{G}
\newcommand{\Ek}{\mathbb{E}_{k}}
\newcommand{\El}{\mathbb{E}_{l}}
\newcommand{\REk}{R^{\Ek \As}}
\newcommand{\Pk}{\mathbb{P}_{k}}
\newcommand{\Mk}{\mathbb{M}_{k}}

\newcommand{\MLck}{\MLk(\Count)}
\newcommand{\eqEMk}{\equiv^{\exists\MLk}}
\newcommand{\eqMk}{\equiv^{\MLk}}
\newcommand{\eqMck}{\equiv^{\MLck}}

\newcommand{\pref}{\sqsubseteq}
\newcommand{\prefgt}{\sqsupseteq}
\newcommand{\Ralph}{R_{\alpha}}
\newcommand{\PMA}{P^{\Mk (\As, a)}}

\newcommand{\PA}{P^{\As}}
\newcommand{\PB}{P^{\Bs}}
\newcommand{\RMA}{R^{\Mk (\As, a)}}

\newcommand{\simord}{\preceq}
\newcommand{\Tk}{\mathbb{P}_k}
\newcommand{\Kcomp}{\bullet}
\newcommand{\nset}{\mathbf{n}}
\newcommand{\kset}{\mathbf{k}}

\newcommand{\Linf}{\mathcal{L}_{\infty,\omega}}
\newcommand{\MLk}{\mathcal{M}_{k}}
\newcommand{\lsem}{\llbracket}
\newcommand{\rsem}{\rrbracket}
\newcommand{\boxa}{\Box_{\alpha}}
\newcommand{\dia}{\Diamond_{\alpha}}
\newcommand{\elen}{\exists_{\leq n}}
\newcommand{\egen}{\exists_{\geq n}}
\newcommand{\eqaEk}{\rightleftarrows_{k}^{\mathbb{E}}}
\newcommand{\eqbEk}{\leftrightarrow_{k}^{\mathbb{E}}}
\newcommand{\eqcEk}{\cong_{k}^{\mathbb{E}}}
\newcommand{\eqaPk}{\rightleftarrows_{k}^{\mathbb{P}}}
\newcommand{\eqbPk}{\leftrightarrow_{k}^{\mathbb{P}}}
\newcommand{\eqcPk}{\cong_{k}^{\mathbb{P}}}
\newcommand{\eqaMk}{\rightleftarrows_{k}^{\mathbb{M}}}
\newcommand{\eqbMk}{\leftrightarrow_{k}^{\mathbb{M}}}
\newcommand{\eqcMk}{\cong_{k}^{\mathbb{M}}}

\newcommand{\gb}{\sim^{\mathsf{g}}}

\newcommand{\Fraisse}{Fra\"{i}ss\'{e}~}
\newcommand{\pow}{\mathscr{P}}
\newcommand{\comp}{{\uparrow}}
\newcommand{\adj}{\frown}
\newcommand{\da}{{\downarrow}}
\newcommand{\hgt}{\mathsf{ht}}
\newcommand{\td}{\mathsf{td}}
\newcommand{\tw}{\mathsf{tw}}
\newcommand{\md}{\mathsf{md}}
\newcommand{\Gf}{\mathcal{G}}
\newcommand{\cnE}{\kappa^{\mathbb{E}}}
\newcommand{\cnP}{\kappa^{\mathbb{P}}}
\newcommand{\cnM}{\kappa^{\mathbb{M}}}
\newcommand{\lbfn}{\lambda}
\newcommand{\pth}{\mathsf{path}}
\newcommand{\labar}[1]{\overset{#1}{\to}}
\newcommand{\ve}{\varepsilon}
\newcommand{\Zp}{\mathbb{Z}^{+}}
\newcommand{\Comon}{\mathsf{Comon}}
\newcommand{\MM}{\mathbb{M}}
\newcommand{\PM}{\mathbb{P}}
\newcommand{\Momega}{\MM_{\omega}}
\newcommand{\es}{\varnothing}
\newcommand{\WAB}{\mathsf{W}_{\As,\Bs}}
\newcommand{\SAB}{\mathcal{S}(\As,\Bs)}
\newcommand{\SBA}{\mathcal{S}(\Bs,\As)}
\newcommand{\cvr}{\prec}
\newcommand{\rcvr}{\succ}

\begin{document}

\title{Relating Structure and Power: Comonadic Semantics for Computational Resources}


\author{Samson Abramsky\thanks{samson.abramsky@cs.ox.ac.uk}~}
\author{Nihil Shah\thanks{nihil@berkeley.edu}}

\affil{Department of Computer Science, University of Oxford}
\date{}
\maketitle

\maketitle

\begin{abstract}
Combinatorial games are widely used in finite model theory, constraint satisfaction, modal logic and concurrency theory to characterize logical equivalences between structures. 
In particular, Ehrenfeucht-\Fraisse games, pebble games, and bisimulation games play a central role. 
We show how each of these types of games can be described in terms of an indexed family of comonads on the category of relational structures and homomorphisms. The index $k$ is a resource parameter which bounds the degree of access to the underlying structure. 
The coKleisli categories for these comonads can be used to give syntax-free characterizations of a wide range of important logical equivalences. Moreover, the coalgebras for these indexed comonads can be used to characterize key combinatorial parameters: tree-depth for the Ehrenfeucht-\Fraisse comonad, tree-width for the pebbling comonad, and synchronization-tree depth for the modal unfolding comonad. These results pave the way for systematic connections between two major branches of the field of logic in computer science which hitherto have been almost disjoint: categorical semantics, and finite and algorithmic model theory.
\end{abstract}

\section{Introduction}

There is a remarkable divide in the field of logic in Computer Science, between two distinct strands: one focussing on semantics and compositionality (``Structure''), the other on expressiveness and complexity (``Power''). 
It is remarkable because these two fundamental aspects of our field are studied using almost disjoint technical languages and methods, by almost disjoint research communities.
We believe that bridging this divide is a major issue in Computer Science, and may hold the key to fundamental advances in the field.

In this paper, we develop a novel approach to relating categorical semantics, which exemplifies the first strand, to finite model theory, which exemplifies the second. It builds on the ideas introduced in \cite{abramsky2017pebbling}, but goes much further, showing clearly that there is a strong and robust connection, which can serve as a basis for many further developments.

\subsection*{The setting}
Relational structures and the homomorphisms between them play a fundamental r\^{o}le in finite model theory, constraint satisfaction and database theory. The existence of a homomorphism $A \rarr B$ is an equivalent formulation of constraint satisfaction, and also equivalent to the preservation of existential positive sentences \cite{chandra1977optimal}. This setting also generalizes what has become a central perspective in graph theory \cite{hell2004graphs}.

\subsection*{Model theory and deception}
In a sense, the purpose of model theory is ``deception''.  It allows us to see structures not ``as they really are'', \ie up to isomorphism, but only up to \emph{definable properties}, where definability is relative to a logical language $\LL$. The key notion is \emph{logical equivalence} $\eqL$. Given structures $\As$, $\Bs$ over the same vocabulary:
\[ \As \eqL \Bs \iffdef \forall \vphi \in \LL. \; \As \models \vphi \; \IFF \; \Bs \models \vphi . \]
If a class of structures $\KK$ is definable in $\LL$, then it must be saturated under $\eqL$. Moreover, for a wide class of cases of interest in finite model theory, the converse holds \cite{kolaitis1992infinitary}.

The idea of syntax-independent characterizations of logical equivalence is quite a classical one in model theory, exemplified by the Keisler-Shelah theorem \cite{shelah1971every}.
It acquires additional significance in finite model theory, where model comparison games such as Ehrenfeucht-\Fraisse games, pebble games and bisimulation games play a central role \cite{Libkin2004}.

We offer a new perspective on these ideas. We shall study these games, not as external artefacts, but as semantic constructions in their own right.
Each model-theoretic comparison game encodes ``deception'' in terms of limited access to the structure. These limitations are indexed by a parameter which quantifies the resources which control this access. For Ehrenfeucht-\Fraisse games and bisimulation games, this is the number of rounds; for pebble games, the number of pebbles.

\subsection*{Main Results}
We now give a conceptual overview of our main results. Technical details will be provided in the following sections. 

We shall consider three forms of model comparison game: Ehrenfeucht-\Fraisse games, pebble games and bisimulation games \cite{Libkin2004}.
For each of these notions of game $\GG$, and value of the resource parameter $k$, we shall define a corresponding \emph{comonad} $\Ck$ on the category of relational structures and homomorphisms  over some relational vocabulary. For each structure $\As$, $\Ck \As$ is another structure over the same vocabulary, which encodes the limited access to $\As$ afforded by playing the game on $\As$ with $k$ resources. There is always an associated homomorphism $\epsA : \Ck \As \rarr \As$ (the \emph{counit} of the comonad), so that $\Ck \As$ ``covers'' $\As$. Moreover, given a homomorphism $h : \Ck \As \rarr \Bs$, there is a \emph{Kleisli coextension} homomorphism $h^* : \Ck \As \rarr \Ck \Bs$. This allows us to form the \emph{coKleisli category} $\Kl(\Ck)$ for the comonad. The objects are relational structures, while the morphisms from $\As$ to $\Bs$ in $\Kl(\Ck)$ are exactly the homomorphisms of the form $\Ck \As \rarr \Bs$. Composition of these morphisms uses the Kleisli coextension. The connection between this construction and the corresponding form of game $\GG$ is expressed  by the following result:
\begin{theorem}
The following are equivalent:
\begin{enumerate}
\item There is a coKleisli morphism $\Ck \As \rarr \Bs$
\item Duplicator has a winning strategy for the existential $\GG$-game with $k$ resources, played from $\As$ to $\Bs$.
\end{enumerate}
\end{theorem}
The existential form of the game has only a ``forth'' aspect, without the ``back''. This means that Spoiler can only play in $\As$, while Duplicator only plays in $\Bs$. This corresponds to the asymmetric form of the coKleisli morphisms $\Ck \As \rarr \Bs$. Intuitively, Spoiler plays in $\Ck \As$, which gives them limited access to $\As$, while Duplicator plays in $\Bs$. The Kleisli coextension guarantees that Duplicator's strategies can always be lifted to $\Ck \Bs$; while we can always compose a strategy $\Ck \As \rarr \Ck \Bs$ with the counit on $\Bs$ to obtain a coKleisli morphism.

This asymmetric form may seem to limit the scope of this approach, but in fact this is not the case. For each of these comonads $\Ck$, we have the following equivalences:
\begin{itemize}
\item $\As \eqak \Bs$ iff there are coKleisli morphisms $\Ck \As \rarr \Bs$ and $\Ck \Bs \rarr \As$. Note that there need be no relationship between these morphisms.
\item $\As \eqck \Bs$ iff $\As$ and $\Bs$ are isomorphic in the coKleisli category $\Kl(\Ck)$. This means that there are morphisms $\Ck \As \rarr \Bs$ and $\Ck \Bs \rarr \As$ which are inverses of each other in $\Kl(\Ck)$.
\end{itemize}
Clearly, $\eqck$ strictly implies $\eqak$.
We can also define an intermediate ``back-and-forth'' equivalence $\eqbk$, parameterized by a winning condition $\WAB \subseteq \Ck \As \times \Ck \Bs$.

For each of our three types of game, there are corresponding  fragments $\Lk$ of first-order logic:
\begin{itemize}
\item For Ehrenfeucht-\Fraisse games, $\Lk$ is the fragment of quantifier-rank $\leq k$.
\item For pebble games, $\Lk$ is the $k$-variable fragment.
\item For bismulation games over relational vocabularies with symbols of arity at most 2, $\Lk$ is the modal fragment \cite{andreka1998modal} with modal depth $\leq k$.
\end{itemize}
In each case, we write $\ELk$ for the existential positive fragment of $\Lk$, and $\Lck$ for the extension of $\Lk$ with counting quantifiers \cite{Libkin2004}.

We can now state our first main result, in a suitably generic form.
\begin{theorem}
For finite structures $\As$ and $\Bs$:

\begin{tabular}{llcl}
(1) & $\As \eqELk \Bs$ & $\; \IFF \;$ & $\As \eqak \Bs$. \\
(2) & $\As \eqLk \Bs$ & $\; \IFF \;$ & $\As \eqbk \Bs$. \\
(3) & $\As \eqLck \Bs$ & $\; \IFF \;$ & $\As \eqck \Bs$.
\end{tabular}
\end{theorem}
Note that this is really a family of three theorems, one for each type of game $\GG$. Thus in each case, we capture the salient logical equivalences in syntax-free, categorical form.

We now turn to the significance of indexing by the resource parameter $k$. When $k \leq l$, we have a natural inclusion morphism $\Ck \As \rarr \Cl \As$, since playing with $k$ resources is a special case of playing with $l \geq k$ resources. This tells us that the smaller $k$ is, the easier it is to find a morphism $\Ck \As \rarr \Bs$. Intuitively, the more we restrict Spoiler's abilities to access the structure of $\As$, the easier it is for Duplicator to win the game.

The contrary analysis applies to morphisms $A \rarr \Ck \Bs$. The smaller $k$ is, the \emph{harder} it is find such a morphism. Note, however, that if $\As$ is a finite structure of cardinality $k$, then $\As \eqak \Ck \As$. In this case, with $k$ resources we can access the whole of $\As$. What can we say when $k$ is strictly smaller than the cardinality of $\As$?

It turns out that there is a beautiful connection between these indexed comonads and combinatorial invariants of structures. This is mediated by the notion of \emph{coalgebra}, another fundamental (and completely general) aspect of comonads. A coalgebra for a comonad $\Ck$ on a structure $\As$ is a morphism $\As \rarr \Ck \As$ satisfying certain properties. We define the \emph{coalgebra number} of a structure $\As$, with respect to the indexed family of comonads $\Ck$, to be the least $k$ such that there is a $\Ck$-coalgebra  on $\As$.

We now come to our second main result.
\begin{theorem}
\begin{itemize}
\item For the pebbling comonad, the coalgebra number of $\As$ corresponds precisely to the \emph{tree-width} of $\As$.
\item For the Ehrenfeucht-\Fraisse comonad, the coalgebra number of $\As$ corresponds precisely to the \emph{tree-depth} of $\As$ \cite{nevsetvril2006tree}.
\item For the modal comonad, the coalgebra number of $\As$ corresponds precisely to the \emph{modal unfolding depth} of $\As$. 
\end{itemize}
\end{theorem}
The main idea behind these results is that coalgebras on $\As$ are in bijective correspondence with decompositions of $\As$ of the appropriate form. We thus obtain categorical characterizations of these key combinatorial parameters.

\section{Game Comonads}

In this section we will define the comonads corresponding to each of the forms of model comparison game we consider.

Firstly, a few notational preliminaries. A relational vocabulary $\sg$ is a set of relation symbols $R$, each with a specified positive integer arity.
A $\sg$-structure $\As$ is given by a set $A$, the universe of the structure, and for each $R$ in $\sg$ with arity $k$, a relation $\RA \subseteq A^k$. A homomorphism $h : \As \rarr \Bs$ is a function $h : A \rarr B$ such that, for each relation symbol $R$ of arity $k$ in $\sg$, for all $a_1, \ldots , a_k$ in $A$:
$\RA(a_1,\ldots , a_k) \IMP \RB(h(a_1), \ldots , h(a_k))$. We write $\CS$ for the category of $\sg$-structures and homomorphisms.

We shall write $\Alk$ for the set of non-empty sequences of length $\leq k$ on a set $A$. We use list notation $[a_1,\ldots , a_j]$ for such sequences, and indicate concatenation by juxtaposition. We write $s \pref t$ for the prefix ordering on sequences. If $s \preford t$, there is a unique $s'$ such that $ss' = t$, which we refer to as the suffix of $s$ in $t$. For each positive integer $n$, we define $\nset := \{ 1, \ldots , n\}$. 

We shall need a few notions on posets. The comparability relation on a poset  $(P, {\leq})$ is $x \comp y$ iff  $x \leq y$ or $y \leq x$.  A chain in a poset $(P, {\leq})$ is  a subset $C \subseteq P$ such that, for all $x, y \in C$, $x \comp y$. A \emph{forest} is a poset $(F, {\leq})$ such that, for all $x \in F$, the  set of predecessors $\da(x) \, := \, \{ y \in F \mid y \leq x\}$ is a finite chain. The height $\hgt(F)$ of a forest $F$ is $\max_{C} | C |$, where $C$ ranges over chains in $F$.

We recall that a comonad $(G, \varepsilon, \delta)$ on a category $\CC$ is given by a functor $G : \CC \rarr \CC$, and natural transformations $\varepsilon : G \Rightarrow I$ (the counit), and $\delta : G \Rightarrow G^2$ (the comultiplication), subject to the conditions that the following diagrams commute, for all objects $A$ of $\CC$:
\begin{center}
\begin{tikzcd}
G A \ar[r, "\delta_A"]  \ar[d, "\delta_A"']
& G G A \ar[d, "G \delta_A"] \\
G G A \ar[r, "\delta_{G A}"]  
& G G G A
\end{tikzcd}
$\qquad \qquad$
\begin{tikzcd}
G A \ar[r, "\delta_A"] \ar[rd, equal] \ar[d, "\delta_A"'] 
& G G A \ar[d, "G \epsA"] \\
G G A \ar[r, "\varepsilon_{GA}"] 
& G A
\end{tikzcd}
\end{center}

An equivalent formulation is  \emph{comonad in Kleisli form} \cite{manes2012algebraic}. This is given by an object map $G$, arrows $\varepsilon_A : GA \rarr A$ for every object $A$ of $\CC$, and a Kleisli coextension operation which takes $f : GA \rarr B$ to $f^* : GA \rarr GB$. These must satisfy the following equations:
\[ \varepsilon_{A}^* = \id_{\Gd A}, \qquad \varepsilon \circ f^* = f, \qquad (g \circ f^*)^* = g^* \circ f^* . \]
We can then extend $G$ to a functor by $\Gd f = (f \circ \varepsilon)^*$; and if we define the comultiplication $\delta : \Gd \Rightarrow \Gd^2$ by $\delta_{A} = \id_{GA}^*$, then $(\Gd, \varepsilon, \delta)$ is a comonad in the standard sense.
Conversely, given a comonad $(\Gd, \varepsilon, \delta)$, we can define the coextension by $f^* = Gf \circ \delta_A$.
This allows us to define the coKleisli category $\Kl(G)$, with objects the same as those of $\CC$, and morphisms from $A$ to $B$ given by the morphisms in $\CC$ of the form $GA \rarr B$. Kleisli composition of $f : GA \rarr B$ with $g : GB \rarr C$ is given by $g \Kcomp f \, := \, g \circ f^*$.

\subsection{The Ehrenfeucht-\Fraisse Comonad}

We shall define a comonad $\Ek$ on $\CS$ for each positive integer $k$. It will be convenient to define $\Ek$ in Kleisli form.
For each structure $\As$, we define a new structure $\Ek \As$, with universe $\Ek A \, := \, \Alk$. We define the map $\epsA : \Ek A \rarr A$ by $\epsA [a_1, \ldots , a_j ] = a_j$. For each relation symbol $R$ of arity $n$, we define $\REk$ to be the set of $n$-tuples $(s_1, \ldots , s_n)$ of sequences which are pairwise comparable in the prefix ordering, and such that $\RA(\epsA s_1, \ldots , \epsA s_n)$. Finally, we define the coextension. Given a homomorphism $f : \Ek \As \rarr \Bs$, we define $f^* : \Alk \rarr \Blk$ by $f^* [a_1, \ldots , a_j ] = [b_1, \ldots , b_j]$, where $b_i = f [a_1, \ldots , a_i]$, $1 \leq i \leq j$.

\begin{proposition}
The triple $(\Ek, \varepsilon, (\cdot)^*)$ is a comonad in Kleisli form.
\end{proposition}

Intuitively, an element of $\Alk$ represents a play in $\As$ of length $\leq k$. A coKleisli morphism $\Ek \As \rarr \Bs$ represents a Duplicator strategy for the existential Ehrenfeucht-\Fraisse game with $k$ rounds, where Spoiler plays only in $\As$, and $b_i = f [a_1, \ldots , a_i]$ represents Duplicator's response in $\Bs$ to the $i$'th move by Spoiler. The winning condition for Duplicator in this game is that, after $k$ rounds have been played, 
the induced relation $\{ (a_i, b_i) \mid 1 \leq i \leq k \}$ is a partial homomorphism from $\As$ to $\Bs$.

These intuitions are confirmed by the following result.
\begin{theorem}
\label{EFgamethm}
The following are equivalent:
\begin{enumerate}
\item There is a homomorphism $\Ek \As \rarr \Bs$.
\item Duplicator has a winning strategy for the existential Ehrenfeucht-\Fraisse game with $k$ rounds, played from $\As$ to $\Bs$.
\end{enumerate}
\end{theorem}

\subsection{The Pebbling Comonad}

We now turn to the case of pebble games. The following construction appeared in \cite{abramsky2017pebbling}.
Given a structure $\As$, we define $\Tk \As$, which will represent plays of the $k$-pebble game on $\As$.\footnote{In \cite{abramsky2017pebbling} we used the notation $\mathbb{T}_k$ for this comonad.} The universe is  $(\kset \times A)^{+}$, the set of finite non-empty sequences of moves $(p, a)$, where $p \in \kset$ is a pebble index, and $a \in A$. We shall use the notation $s = [ (p_1, a_1), \ldots , (p_n, a_n)]$ for these sequences, which may be of arbitrary length. Thus the universe of $\Tk \As$ is always infinite, even if $\As$ is a finite structure. This is unavoidable, by 
\cite[Theorem 7]{abramsky2017pebbling}. We  define $\epsA : \Tk A \rarr A$ to send a play  $[ (p_1, a_1), \ldots , (p_n, a_n)]$ to $a_n$, the $A$-component of its last move.

Given an $n$-ary relation $R \in \sg$, we define $R^{\Tk \As}(s_1, \ldots , s_n)$ iff (1) the $s_i$ are pairwise comparable in the prefix ordering;
(2) the pebble index of the last move in each $s_i$ does not appear in the suffix of $s_i$ in $s_j$ for any $s_j \prefgt s_i$; and (3) $R^{A}(\epsA(s_1), \ldots , \epsA(s_n))$.

Finally, given a homomorphism $f : \Tk \As \rarr \Bs$, we define $f^* : \Tk A \rarr \Tk B$ by \\
$f^*  [ (p_1, a_1), \ldots , (p_j, a_j)] = [(p_1, b_1), \ldots , (p_j, b_j)]$,
where $b_i = f [(p_1, a_1), \ldots , (p_i, a_i)]$, $1 \leq i \leq j$.

\begin{proposition}
The triple $(\Tk, \varepsilon, (\cdot)^*)$ is a comonad in Kleisli form.
\end{proposition}

The following is  \cite[Theorem 13]{abramsky2017pebbling}.
\begin{theorem}
\label{strmorth}
The following are equivalent:
\begin{enumerate}
\item There is a homomorphism $\Tk \As \rarr \Bs$.
\item There is a winning strategy for Duplicator in the existential $k$-pebble game from $\As$ to $\Bs$.
\end{enumerate}
\end{theorem}

\subsection{The Modal Comonad}
 
For the modal case, we assume that the relational vocabulary $\sg$ contains only symbols of arity at most 2. We can thus regard a $\sigma$-structure as a Kripke structure for a multi-modal logic, where the universe is thought of as a set of worlds, each binary relation symbol $\Ralph$ gives the accessibility relation for one of the modalities, and each unary relation symbol $P$ give the valuation for a corresponding propositional variable. If there are no unary symbols, such structures are exactly the labelled transition systems widely studied in concurrency \cite{milner1989communication}.

Modal logic localizes its notion of satisfaction in a structure to a world. We shall reflect this by using the category of \emph{pointed relational structures} $\CSp$. Objects of this category are pairs $(\As, a)$ where $\As$ is a $\sg$-structure and $a \in A$. Morphisms $h : (\As, a) \rarr (\Bs, b)$ are homomorphisms $h : \As \rarr \Bs$ such that $h(a) = b$. Of course, the same effect could be achieved by expanding the vocabulary $\sg$ with a constant, but pointed categories appear in many mathematical contexts.

For each $k >0$, we shall define a comonad $\Mk$, where $\Mk (\As, a)$ corresponds to unravelling the structure $\As$, starting from $a$, to depth $k$. 
The universe of $\Mk (\As, a)$ comprises the unit sequence $[a]$, which is the distinguished element,  together with all sequences of the form $[a_0, \alpha_1, a_1, \ldots , \alpha_j, a_{j}]$, where $a = a_0$, $1 \leq j  \leq k$, and $\RA_{\alpha_i}(a_i, a_{i+1})$, $0 \leq i < j$. The map $\epsA : \Mk (A, a) \rarr  (A, a)$ sends a sequence to its last element. Unary relation symbols $P$ are interpreted by $\PMA(s)$ iff $\PA(\epsA s)$. For binary relations $\Ralph$, the interpretation is $\RMA_{\alpha}(s,t)$ iff for some $a' \in A$, $t = s[\alpha, a']$.
 Given a morphism $f : \Mk (\As, a) \rarr (\Bs, b)$, we define $f^* : \Mk (\As, a) \rarr \Mk (\Bs, b)$ recursively by $f^*[a] = [b]$,
 $f^*(s[\alpha, a']) = f^*(s)[\alpha, b']$ where $b' = f(s[\alpha, a'])$.
 This is well-defined since $f$ is a morphism by assumption.

\begin{proposition}
The triple $(\Mk, \varepsilon, (\cdot)^*)$ is a comonad in Kleisli form on $\CSp$.
\end{proposition}


We recall the notion of \emph{simulation} between Kripke structures \cite{blackburn2002modal}. Given structures $\As$, $\Bs$, we define relations $\simord_k \; \subseteq \; A \times B$, $k \geq 0$, by induction on $k$: $\simord_0 \; = A\;  \times B$, and $a \simord_{k+1} b$ iff (1) for all unary $P$, $\PA(a)$ implies $\PB(b)$, and (2) for all $\Ralph$, if $\RaA(a, a')$, then for some $b'$, $\RaB(b, b')$ and $a' \simord_k b'$.
It is standard that these relations are equivalently formulated in terms of a modified existential Ehrenfeucht-\Fraisse game \cite{blackburn2002modal,gradel2014freedoms}.

\begin{theorem}
\label{simthm}
Let $\As$, $\Bs$ be Kripke structures, with $a \in A$ and $b \in B$, and $k>0$.
The following are equivalent:
\begin{enumerate}
\item There is a homomorphism $f : \Mk (\As, a) \rarr (\Bs, b)$.
\item $a \simord_k b$.
\item There is a winning strategy for Duplicator in the $k$-round simulation game from $(\As, a)$ to $(\Bs, b)$.
\end{enumerate}
\end{theorem}

\section{Logical Equivalences}

We now show how our game comonads can be used to give syntax-free characterizations of a range of logical equivalences, which play a central r\^ole in finite model theory and modal logic.

We shall be considering logics $\LL$ which arise as fragments of $\Linf$, the extension of first-order logic with infinitary conjunctions and disjunctions, but where formulas contain only finitely many variables. In particular, we will consider the fragments $\Lk$, of formulas with quantifier rank $\leq k$, and $\Lvk$, the $k$-variable fragment. These play a fundamental r\^ole in finite model theory.

We shall also consider two variants for each of these fragments $\LL$. One is the existential positive fragment $\exists\LL$, which contains only those formulas of $\LL$ built using existential quantifiers, conjunction and disjunction. The other is $\Lc$, the extension of $\LL$ with counting quantifiers. These have the form $\elen$, $\egen$, where the semantics of $\As \models \egen x. \, \psi$ is that there exist at least $n$ distinct elements of $A$ satisying $\psi$.

Each of these logics $\LL$ induces an equivalence on structures in $\CS$:
\[ \As \eqL \Bs \iffdef \forall \vphi \in \LL. \; \As \models \vphi \; \IFF \; \Bs \models \vphi . \]
Our aim is to characterize these equivalences in terms of our game comonads, and more specifically, to use morphisms in the coKleisli categories as witnesses for these equivalences.

Two equivalences can be defined uniformly for any indexed family of comonads $\Ck$:
\begin{itemize}
\item $\As \eqaCk \Bs$ iff there are coKleisli morphisms $\Ck \As \rarr \Bs$ and $\Ck \Bs \rarr \As$. Note that there need be no relationship between these morphisms.
This is simply the equivalence induced by the preorder collapse of the coKleisli category.
\item $\As \eqcCk \Bs$ iff $\As$ and $\Bs$ are isomorphic in the coKleisli category $\Kl(\Ck)$. This means that there are morphisms $\Ck \As \rarr \Bs$ and $\Ck \Bs \rarr \As$ which are inverses of each other in $\Kl(\Ck)$.
\end{itemize}
Clearly, $\eqcCk$ strictly implies $\eqaCk$.

We shall also define an intermediate, ``back-and-forth'' equivalence $\eqbCk$. This will be more specific to ``game comonads'' defined on a concrete category such as $\CS$, but it will still be defined and shown to have the appropriate properties in considerable generality.
We assume that for each structure $\As$, the universe $\Ck A$ has a forest order $\preford$, as seen in our concrete constructions using the prefix ordering on sequences.
We add a root $\bot$ for convenience.
We write the covering relation for this order as $\cvr$; thus $s \cvr t$ iff $s \preford t$, $s \neq t$, and for all $u$, $s \preford u \preford t$ implies $u = s$ or $u = t$.
We shall also assume that, for any coKleisli morphism $f : \Ck \As \to \Bs$, the Kleisli coextension preserves the covering relation: $s \cvr s'$ implies $f^*(s) \cvr f^*(s')$.

The definition will be parameterized on a set $\WAB \, \subseteq \, \Ck A \times \Ck B$ of ``winning positions'' for each pair of structures $\As$, $\Bs$.
We assume that a function $f : \Ck A \to B$ such that, for all $s \in \Ck A$, $(s, f^*(s)) \in \WAB$, is a coKleisli morphism.

We define the back-and-forth $\Ck$ game between $\As$ and $\Bs$ as follows.

At the start of each round of the game, the position is specified by $(s, t) \in \Ck A \times \Ck B$. The initial position is $(\bot, \bot)$.  The round proceeds as follows. Either Spoiler chooses some $s' \rcvr s$, and Duplicator responds with $t' \rcvr t$, resulting in a new position $(s', t')$; or Spoiler chooses some $t'' \rcvr t$ and Duplicator responds with $s'' \rcvr s$, resulting in $(s'',t'')$. Duplicator wins the round if the new position is in $\WAB$.

We can then define $\SAB$ to be the set of all functions $f : \Ck A \to B$ such that, for all $s \in \Ck A$, $(s, f^*(s)) \in \WAB$.

We define a \emph{locally invertible pair} $(F, G)$ from $\As$ to $\Bs$ to be a pair of sets $F \subseteq \SAB$, $G \subseteq \SBA$, satisfying the following conditions:
\begin{enumerate}
\item For all $f \in F$, $s \in \Ck A$, for some $g \in G$, $g^* f^*(s) = s$.
\item For all $g \in G$, $t \in \Ck B$, for some $f \in F$, $f^* g^*(t) = t$.
\end{enumerate}
We define $\As \eqbCk \Bs$ iff there is a non-empty locally invertible pair from $\As$ to $\Bs$.

\begin{proposition}
\label{pisoprop}
The following are equivalent:
\begin{enumerate}
\item $\As \eqbCk \Bs$.
\item There is a winning strategy for Duplicator in the $\Ck$ game between $\As$ and $\Bs$.
\end{enumerate}
\end{proposition}
\begin{proof} 
Assuming (1), with a locally invertible pair $(F, G)$, we define a strategy for Duplicator inductively,
such that after each round, the play is within the set
\[ \{ (s, f^*(s)) \mid s \in \Ck A, f \in F \} \; = \; \{ (g^*(t), t) \mid t \in \Ck B, g \in G \} . \]
 Assume $(s, t)$ has been played. If Spoiler now plays $s' \rcvr s$ in $\Ck A$, then there is  $f \in F$ such that $f^*(s) = t$, and we respond with $t' = f^*(s') \rcvr f^*(s)$.
 Since $f \in \SAB$, $(s',t') \in \WAB$. The case when Spoiler plays in $\Ck B$ is symmetric.

Assuming (2),
let  $\Phi$ be the set of all plays $(s, t)$ following the Duplicator strategy. 
Define 
\begin{align*}
F \, & := \, \{ f : \Ck A \rarr B \mid \forall s \in \Ck A. \, (s, f^*(s))  \in \Phi  \}, \\
G \, & := \, \{ g : \Ck B \rarr A \mid \forall t \in \Ck B. \, (g^*(t), t) \in \Phi \} . 
\end{align*}
Since the strategy is winning, $\Phi \subseteq \WAB$, and $F \subseteq \SAB$, $G \subseteq \SBA$.
We claim that for all $(s, t) \in \Phi$: (A) $\exists f \in F. \, f^*(s) = t$, and (B) $\exists g \in G. \, g^*(t) = s$. (A) follows by extending $(s, t)$ to a morphism $f : \Ck \As \rarr \Bs$. For any $s' \preford s$, we assign the corresponding predecessor of $t$.
For any $s'$ which is not a predecessor of $s$, let $s_1 = s \sqcap s'$, the meet of $s$ and $s'$.
We write $t_1$ for the corresponding predecessor of $t$.
We define $f$ on $s'$ by assigning $t_1$ in response to $s_1$, and then following Duplicator's responses as Spoiler plays according to $s'$ in $\Ck A$.
(B) follows by a symmetric argument.

Now for any $f \in F$ and $s \in \Ck A$, $(s, f^*(s)) \in \Phi$, and hence by (B) we can find $g \in G$ to witness local invertibility; the case for $g \in G$ and $t \in \Ck B$ is symmetric.
\end{proof}

The local invertibility condition on a pair of sets $(F, G)$ has a fixpoint characterization, which may be of some interest. 
We define set functions $\Gamma : \pow(\SAB) \rarr \pow(\SBA)$, $\Delta : \pow(\SBA) \rarr \pow(\SAB)$:  
\begin{align*}
\Gamma(F) & = \{ g \in T \mid \forall t \in \Ck B. \exists f \in F. \, f^* g^* t = t \}, \\
\Delta(G) & = \{ f \in S \mid \forall s \in \Ck A. \exists g \in G. \,  g^* f^* s = s \} .
\end{align*}
These functions are monotone. Moreover, a pair of sets $(F, G)$ is locally invertible iff $F \subseteq \Delta(G)$ and $G \subseteq \Gamma(F)$. These conditions in turn imply that $F \subseteq \Delta \Gamma(F)$, and if this holds, then we can set $G := \Gamma(F)$ to obtain a locally invertible pair $(F, G)$. Thus existence of a locally invertible pair is equivalent to the existence of non-empty $F$ such that $F \subseteq \Theta(F)$, where $\Theta = \Delta \Gamma$. Since $\Theta$ is monotone, by Knaster-Tarski this is equivalent to the greatest fixpoint of $\Theta$ being non-empty. (Note that $\Theta(\varnothing) = \varnothing$).

If $\As$ and $\Bs$ are finite, so is $S$, and we can construct the greatest fixpoint by a finite descending sequence
$S \supseteq \Theta(S) \supseteq \Theta^2(S) \supseteq \cdots$. This fixpoint is non-empty iff $\As \eqbEk \Bs$.

We shall now turn to a detailed study of each of our comonads in turn.

\subsection{The Ehrenfecht-\Fraisse comonad}

A coKelisli morphism $f : \Ek \As \to \Bs$ is an $I$-morphism if $s \preford t$ and $\epsA(s) = \epsA(t)$ implies that $f(s) = f(t)$. An equivalent statement is that, if we add a binary relation symbol $I$ to the vocabulary, and set $I^{\As}$ to be the identity relation on $A$, and $I^{\Bs}$ to be the identity relation on $B$, then $f$ is also a homomorphism with respect to $I$. The significance of this condition is that, if $f : \Ek \As \to \Bs$ and $g : \Ek \Bs \to \As$ are $I$-morphisms, then $f^*(s) =t$, $g^*(t) = s$ imply that $(s,t)$ defines a partial isomorphism from $\As$ to $\Bs$. We refine the definition of the equivalence $\eqcEk$ accordingly. We say that $\As \eqcEk \Bs$ iff there are $I$-morphisms $f : \Ek \As \to \Bs$ and $g : \Ek \Bs \to \As$ with ${f^*}^{-1} = g^*$.

Note that, for any coKleisli morphism $f : \Ek \As \to \Bs$, there is an $I$-morphism $f_{I} : \Ek \As \to \Bs$, obtained by firstly restricting $f$ to non-repeating sequences, then extending it by applying the $I$-morphism condition for repetitions. It is easy to verify that $f_{I}$ is a homomorphism. Thus there is no need to modify the equivalence $\eqaEk$.

We define $\WAB^{\Ek}$ to be the set of pairs $(s, t) \in \Ek \As \times \Ek \Bs$ such that $s = [a_1, \ldots , a_j]$, $t = [b_1, \ldots , b_j]$, and $\{ (a_i, b_i) \mid 1 \leq i \leq j \}$ defines a partial isomorphism from $\As$ to $\Bs$. This specifies the back-and-forth equivalence $\eqbEk$.

We now recall the \emph{bijection game} \cite{Hella1996}. In this variant of the Ehrenfeuch-\Fraisse game, Spoiler wins if the two structures have different cardinality. Otherwise, at round $i$, Duplicator chooses a bijection $\psi_i$ between $A$ and $B$, and Spoiler chooses an element $a_i$ of $A$. This determines the choice by Duplicator of $b_i = \psi_i(a_i)$. Duplicator wins after $k$ rounds if the relation $\{ (a_i, b_i) \mid 1 \leq i \leq k \}$ is a partial isomorphism.

\begin{proposition}
\label{bijgameprop}
The following are equivalent, for finite structures $\As$ and $\Bs$:
\begin{enumerate}
\item $\As \eqcEk \Bs$.
\item There is a winning strategy for Duplicator in the $k$-round bijection game.
\end{enumerate}
\end{proposition}
\begin{proof}
Assuming (1), we have $I$-morphisms $f : \Ek \As \to \Bs$ and $g : \Ek \Bs \to \As$ with $g^* = {f^*}^{-1}$. 
For each $s \in \{ []\} \cup A^{<k}$, we can define a map $\psi_s : A \to B$, by $\psi_s (a) = f(s[a])$. This is a bijection, with inverse defined similarly from $g$. These bijections provide a strategy for Duplicator. Since each $(s, f^*(s))$ is a partial isomorphism, this is a winning strategy.

Conversely, a winning strategy provides bjiections $\psi_s$, which we can use to define $f$ by $f(s[a]) = \psi_s(a)$. The winning conditions imply that this is an $I$-isomorphism in the coKleisli category.
\end{proof}


We can now state our main result on logical equivalences for the Ehrenfeucht-\Fraisse comonad.
\begin{theorem}
\begin{enumerate}
\item For all structures $\As$ and $\Bs$: $\As  \eqELk \Bs \; \IFF \; \As \eqaEk \Bs$.
\item  For all structures $\As$ and $\Bs$: $\As \eqLk \Bs \; \IFF \; \As \eqbEk \Bs$.
\item For all finite structures $\As$ and $\Bs$: $\As \eqLck \Bs \; \IFF \; \As \eqcEk \Bs$.
\end{enumerate}
\label{thm:mainEF}
\end{theorem}
\begin{proof}
(1) follows from Theorem~\ref{EFgamethm} and standard results \cite{kolaitis1990expressive}.
(2) follows from Proposition~\ref{pisoprop} and the Ehrenfeucht-\Fraisse theorem \cite{ebbinghaus2005finite}.
(3) follows from Proposition~\ref{bijgameprop} and results originating in \cite{Hella1996} and expounded in \cite{Libkin2004}.
\end{proof}

If we modify $\WAB^{\Ek}$, and hence $\eqbEk$, by asking for partial correspondences rather than partial isomorphisms, we obtain a characterization of elementary equivalence for equality-free logic \cite{Casanovas1996}.

\subsection{The Pebbling Comonad}

A similar notion of $I$-morphism applies to the pebbling comonad as we saw previously with the Ehrenfeucht-\Fraisse comonad \cite{abramsky2017pebbling}. 

Given $s = [(p_1,a_1), \ldots , (p_n,a_n)] \in \Pk \As$ and $t = [(p_1,b_1), \ldots , (p_n,b_n)] \in \Pk \Bs$, we define $\phi_{s, t} \; = \; \{ (a_p, b_p) \mid p \in \kset, \; \mbox{$p$ occurs in $s$} \}$, where the last occurrence of $p$ in $s$ is on $a_p$, and the corresponding last occurrence in $t$ is on $b_p$.
We define $\WAB^{\Pk}$ to be the set of all such $(s,t)$ for which $\phi_{s,t}$ is a partial isomorphism. This specifies the back-and-forth equivalence $\eqbPk$.

We now state the following result, characterizing the equivalences induced by finite-variable logics $\Lvk$.
\begin{theorem}
\begin{enumerate}
\item For all structures $\As$ and $\Bs$: $\As  \eqELvk \Bs \; \IFF \; \As \eqaPk \Bs$.
\item  For all finite structures $\As$ and $\Bs$: $\As \eqLvk \Bs \; \IFF \; \As \eqbPk \Bs$.
\item For all finite structures $\As$ and $\Bs$: $\As \eqLvck \Bs \; \IFF \; \As \eqcPk \Bs$.
\end{enumerate}
\label{thm:mainPebble}
\end{theorem}
\begin{proof}
This follows from Theorems 14, 18 and 20 of \cite{abramsky2017pebbling}.
\end{proof}

\subsection{The Modal Comonad}

The key notion of equivalence in modal logic is bisimulation \cite{blackburn2002modal,sangiorgi2009origins}. We shall define the finite approximants to bisimulation \cite{hennessy1980observing}.\footnote{Our focus on finite approximants in this paper is for uniformity, and because they are relevant in resource terms. We can extend the comonadic semantics beyond the finite levels. We shall return to this point in the final section.}
Given Kripke structures $\As$ and $\Bs$, we define a family of relations $\sim_k \; \subseteq \; A \times B$: $\sim_0 \, = \, A \times B$; $a \sim_{k+1} b$ iff (1) for all unary $P$, $\PA(a)$ iff $\PB(b)$; and (2) for all binary $\Ralph$, $\RaA(a, a')$ implies for some $b'$, $\RaB(b,b')$ and $a' \sim_k b'$, and $\RaB(b,b')$ implies for some $a'$, $\RaA(a, a')$ and $a' \sim_k b'$.

We define $\WAB^{\Mk}$ to be the set of all $(s, t) \in \Mk (\As, a) \times \Mk (\Bs, b)$ such that 
\[s = [a_0, \alpha_1, a_1, \ldots , \alpha_j, a_{j}], \qquad t = [b_0, \alpha_1, b_1, \ldots , \alpha_j, b_{j}], \]
and for all $i$ and all unary $P$, $\PA(a_i)$ iff $\PB(b_i)$. This specifies the back-and-forth equivalence $\eqbMk$.

\begin{theorem}
\label{bisimthm}
For pointed Kripke structures $(\As, a)$ and $(\Bs, b)$: $a \sim_k b$ iff $(\As, a) \eqbMk (\Bs, b)$.
\end{theorem}

Turning to logic, we will  consider $\MLk$, the \emph{modal fragment} of modal depth $\leq k$. This arises from the standard translation of (multi)modal logic into $\Linf$ \cite{blackburn2002modal}. Let us fix a relational vocabulary $\sg$ with symbols of arity $\leq 2$. For each unary symbol $P$, there will be a corresponding propositional variable $p$. Formulas are built from these propositional variables by propositional connectives, and modalities $\boxa$, $\dia$ corresponding to the binary relation symbols $\Ralph$ in $\sigma$. Modal formulas $\vphi$ then admit a translation into formulas $\lsem \vphi \rsem = \psi(x)$ in one free variable. The translation sends propositional variables $p$ to $P(x)$, commutes with the propositional connectives, and  sends $\dia \vphi$ to $\exists y. \, \Ralph(x,y) \AND \psi(y)$, where $\psi(x) = \lsem \vphi \rsem$. This translation is semantics-preserving: given  a $\sg$-structure $\As$ and $a \in A$, then $\As, a \models \vphi$ in the sense of Kripke semantics iff $\As \models \psi(a)$ in the standard model-theoretic sense, where $\psi(x) = \lsem \vphi \rsem$.

We define the modal depth of a modal formula $\vphi$ as the maximum nesting depth of modalities occurring in $\vphi$. $\MLk$ is then the image of the translation of modal formulas of modal depth $\leq k$. The existential positive fragment $\exists\MLk$ arises from the modal sublanguage in which formulas are built from propositional variables using only conjunction, disjunction and the diamond modalities $\dia$.

Extensions of the modal language with counting capabilities have been studied in the form of \emph{graded modalities} \cite{Rijke2000}. These have the form $\dia^n$, $\boxa^n$, where $\As, a \models \dia^n \vphi$ if there are at least $n$ $\Ralph$-successors of $a$ which satisfy $\vphi$. 
We define $\MLk(\Count)$ to be the extension of the modal fragment with graded modalities.

A corresponding notion of graded bisimulation is given in  \cite{Rijke2000}. This is in turn related to \emph{resource bismulation} \cite{corradini1999graded}, which has been introduced in the concurrency setting. The two notions are shown to coincide for image-finite Kripke structures in \cite{aceto2010resource}, who also show that they can be presented in a simplified form.
We recall that a Kripke structure $\As$ is image-finite if for all $a \in A$ and $\Ralph$, $\Ralph(a) \, := \, \{ a' \mid \RA(a,a') \}$  is finite.

Adapting the results in \cite{aceto2010resource}, we define approximants $\gb_k$ for graded bisimulation: $\gb_0 \, = \, A \times B$, and $a \gb_{k+1} b$ if for all $P$, $\PA(a)$ iff $\PB(b)$, and for all $\Ralph$, there is a bijection $\theta : \RA(a) \cong \RB(b)$ such that, for all $a' \in \RA(a)$, $a' \gb_k \theta(a')$.

We can also define a corresponding graded bisimulation game between $(\As, a)$ and $(\Bs, b)$. At round 0, the elements $a_0 = a$ and $b_0 = b$ are chosen. Duplicator wins if for all $P$, $\PA(a)$ iff $\PB(b)$, otherwise Spoiler wins. At round $i+1$, Spoiler chooses some $\Ralph$, and Duplicator chooses a bijection $\theta_i : \RaA(a_i) \cong \RaB(b_i)$. If there is no such bijection, Spoiler wins. Otherwise, Spoiler then chooses $a_{i+1} \in \RA(a_i)$, and $b_{i+1} := \theta_i(a_{i+1})$. Duplicator wins this round if for all $P$, $\PA(a_{i+1})$ iff $\PB(b_{i+1})$, otherwise Spoiler wins.

This game is evidently analogous to the bijection game we encountered previously.
\begin{proposition}
\label{gradedbisimprop}
The following are equivalent:
\begin{enumerate}
\item There is a winning strategy for Duplicator in the $k$-round graded bisimulation game between $(\As, a)$ and $(\Bs, b)$.
\item $a \gb_k b$.
\item $(\As, a) \eqcMk (\Bs, b)$.
\end{enumerate}
\end{proposition}

\begin{theorem}
\begin{enumerate}
\item For all Kripke structures $\As$ and $\Bs$: $\As  \eqEMk \Bs \; \IFF \; \As \eqaMk \Bs$.
\item  For all Kripke structures $\As$ and $\Bs$: $\As \eqMk \Bs \; \IFF \; \As \eqbMk \Bs$.
\item For all image-finite Kripke structures $\As$ and $\Bs$: $\As \eqMck \Bs \; \IFF \; \As \eqcMk \Bs$.
\end{enumerate}
\end{theorem}
\begin{proof}
(1) follows from Proposition~\ref{simthm} and standard results on preservation of existential positive modal formulas by simulations \cite{blackburn2002modal}.
(2) follows from Theorem~\ref{bisimthm} and the  Hennesy-Milner Theorem \cite{hennessy1980observing,blackburn2002modal}.
(3) follows from Proposition~\ref{gradedbisimprop} and the results in  \cite{Rijke2000,aceto2010resource}.
\end{proof}

\section{Coalgebras and combinatorial parameters}
Another fundamental aspect of comonads is that they have an associated notion of \emph{coalgebra}. A coalgebra for a comonad $(G, \varepsilon, \delta)$ is a morphism $\alpha : A \to G A$ such that the following diagrams commute:
\begin{center}
\begin{tikzcd}
A  \ar[r, "\alpha"] \ar[d, "\alpha"']
& G A \ar[d,  "\delta_{A}"] \\
G A  \ar[r, "G \alpha"] 
& G^2 A
\end{tikzcd}  
$\qquad \qquad$
\begin{tikzcd}
A \ar[r, "\alpha"] \ar[rd, "\id_A"']
& G A \ar[d, "\epsA"] \\
& A
\end{tikzcd}
\end{center}

Our use of indexed comonads $\Ck$ opens up a new kind of question for coalgebras. Given a structure $\As$, we can ask: what is the least value of $k$ such that a $\Ck$-coalgebra exists on $\As$?  We call this the \emph{coalgebra number} of $\As$. We shall find that for each of our comonads, the coalgebra number is a significant combinatorial parameter of the structure.

\subsection{The Ehrenfeucht-\Fraisse comonad and tree-depth}

A graph is $G = (V, {\adj})$, where $V$ is the set of vertices, and $\adj$ is the adjacency relation, which is symmetric and irreflexive.
A forest cover for $G$ is a forest $(F, {\leq})$ such that $V \subseteq F$, and if $v \adj v'$, then $v \comp v'$.
The tree-depth $\td(G)$ is defined to be $\min_{F} \hgt(F)$, where $F$ ranges over forest covers of $G$.\footnote{We formulate this notion in order-theoretic rather than graph-theoretic language, but it is equivalent to the definition in \cite{nevsetvril2006tree}.} It is clear that we can restrict  to forest covers of the form $(V, {\leq})$, since given a forest cover $(F, {\leq})$ of $G = (V, {\adj})$, $(V, \, {\leq} \cap V^2)$ is also a forest cover of $G$, and $\hgt(V) \leq \hgt(F)$. Henceforth, by forest covers of $G$ we shall mean those with universe $V$.

Given a $\sg$-structure $\As$, the Gaifman graph $\Gf(\As)$ is $(A, \adj)$, where $a \adj a'$ iff for some relation $R \in \sg$, for some $(a_1, \ldots , a_n) \in \RA$, $a = a_i$, $a' = a_j$, $a \neq a'$. The tree-depth of $\As$ is $\td(\Gf(\As))$.

\begin{theorem}
Let $\As$ be a finite $\sg$-structure, and $k>0$. There is a bijective correspondence between
\begin{enumerate}
\item $\Ek$-coalgebras $\alpha : \As \rarr \Ek \As$.
\item Forest covers of $\Gf(\As)$ of height $\leq k$.
\end{enumerate}
\end{theorem}
\begin{proof}
Suppose that $\alpha : \As \to \Ek \As$ is a coalgebra. For $a \in A$, let $\alpha(a) = [a_1, \ldots , a_j]$. The first coalgebra equation says that $\alpha(a_i) = [a_1, \ldots , a_i]$, $1 \leq i \leq j$. The second says that $a_j = a$. Thus $\alpha : A \to \Alk$ is an injective map whose image is a prefix-closed subset of $\Alk$. Defining $a \leq a'$ iff $\alpha(a) \preford \alpha(a')$ yields a forest order on $A$, of height $\leq k$. If $a \adj a'$ in $\Gf(\As)$, for some $a_1, \ldots , a_n$ with $a = a_i$, $a' = a_j$, we have $\RA(a_1, \ldots , a_n)$. Since $\alpha$ is a homomorphism, we must have $R^{\Ek \As}(\alpha(a_1), \ldots , \alpha(a_n))$, hence $\alpha(a_i) \comp \alpha(a_j)$, and so $a_i \comp a_j$. Thus $(A, {\leq})$ is a forest cover of $\As$, of height $\leq k$.

Conversely, given such a forest cover $(A, {\leq})$, for each $a \in A$, its predecessors form a chain $a_1 < \cdots < a_j$, with $a_j = a$, and $j \leq k$.
We define $\alpha(a) = [a_1, \ldots , a_j]$, which yields a map $\alpha : A \to \Alk$, which evidently satisfies the coalgebra equations. If $\RA(a_1, \ldots , a_n)$, then since $(A, {\leq})$ is a forest cover, we must have $a_i \comp a_j$ for all $i, j$, and hence $\alpha(a_i) \comp \alpha(a_j)$. Thus $\alpha$ is a homomorphism.
\end{proof}

\noindent We write $\cnE(\As)$ for the coalgebra number of $\As$ with respect to the the Ehrenfeucht-\Fraisse comonad.

\begin{theorem}
For all finite structures $\As$: $\td(\As) \, = \, \cnE(\As)$.
\end{theorem}

\subsection{The pebbling comonad and tree-width}

We review the notions of tree decompositions and tree-width. A tree $(T, {\leq})$ is a forest with a least element (the root). A tree is easily seen to be a meet-semilattice: every pair of elements $x, x'$ has a greatest lower bound $x \wedge x'$ (the greatest common ancestor). The path from $x$ to $x'$ is  the set
$\pth(x, x') := [x \wedge x', x] \cup [x \wedge x', x']$, where we use interval notation: $[y, y'] := \{ z \in T \mid y \leq z \leq y' \}$. 

A tree-decomposition of a graph $G = (V, {\adj})$ is a tree $(T, {\leq})$ together with a labelling function $\lbfn : T \rarr \pow(V)$ satisfying the following conditions: 
\begin{itemize}
\item (TD1) for all $v \in V$, for some $x \in T$, $v \in \lbfn(x)$; 
\item (TD2) if $v \adj v'$, then for some $x \in T$, $\{ v, v' \} \subseteq \lbfn(x)$; 
\item (TD3) if $v \in \lbfn(x) \cap \lbfn(x')$, then for all $y \in \pth(x, x')$, $v \in \lbfn(y)$. 
\end{itemize}
The width of a tree decomposition is given by $\max_{x \in T} |\lbfn(x)| -1$. We define the tree-width  $\tw(G)$ of a graph $G$ as $\min_{T} \mathsf{width}(T)$, where $T$ ranges over tree decompositions of $G$.

We shall now give an alternative formulation of tree-width which will provide a useful bridge to the coalgebraic characterization. It is also interesting in its own right: it clarifies the relationship between tree-width and tree-depth, and shows how pebbling arises naturally in connection with tree-width.

A $k$-pebble forest cover for a graph $G = (V, {\adj})$ is a forest cover $(V, {\leq})$ together with a pebbling function $p : V \to \kset$ such that, if $v \adj v'$ with $v \leq v'$, then for all $w \in (v,v']$, $p(v) \neq p(w)$.

The following result is implicit in \cite{abramsky2017pebbling}, but it seems worthwhile to set it out more clearly.

\begin{theorem}
Let $G$ be a finite graph. The following are equivalent:
\begin{enumerate}
\item $G$ has a tree decomposition of width $< k$.
\item $G$ has a $k$-pebble forest cover.
\end{enumerate}
\end{theorem}
\begin{proof}
$(1) \Rightarrow (2)$. Assume that $G = (V, \adj)$ has a tree decomposition $(T, {\leq}, \lambda)$ of width $< k$. We say that a tree decomposition is \emph{orderly} if it has the following property: for all $x \in T$, there is at most one $v \in \lambda(x)$ such that for all $x' < x$, $v \not\in \lambda(x')$.
Nice tree decompositions are orderly \cite{kloks1994treewidth}; hence by standard results, without loss of generality we can assume that the given tree decomposition is orderly.

For any $v \in V$, the set of $x \in T$ such that $v \in \lambda(x)$ is non-empty by (TD1), and closed under meets by (TD3). Since $T$ is a tree, this implies that this set has a least element $\tau(v)$. This defines a function $\tau : V \to T$. The fact that tree decomposition is orderly implies that $\tau$ is injective. We can define an order on $V$ by $v \leq v'$ iff $\tau(v) \leq \tau(v')$. This is isomorphic to a sub-poset of $T$, and hence is a forest order.

We define $p : V \to \kset$ by induction on this order. Assuming $p(v')$ is defined for all $v' < v$, we consider $\tau(v)$. Since the tree decomposition is orderly, this means in particular that $p(v')$ is defined for all $v' \in S  :=  \lambda(\tau(v)) \setminus \{v\}$. Since the decomposition is of width $< k$, we must have $|S| < k$. We set $p(v) := \min (\kset \setminus \{ p(v') \mid v' \in S \})$.

To verify that $(V, {\leq})$ is a forest cover, suppose that $v \adj v'$. By (TD2), for some $x \in T$, $\{ v, v' \} \subseteq \lambda(x)$. We have $\tau(v) \leq x \geq \tau(v')$, and since $T$ is a tree, we must have $\tau(v) \, \comp \, \tau(v')$, whence $v \, \comp \, v'$.

Finally, we must verify the condition on the pebbling function $p$. Suppose that $v \adj v'$, and $v < w \leq v'$. 
Since $v \adj v'$, for some $x$, $\{ v, v' \} \subseteq \lambda(x)$. But then $\tau(v) < \tau(w) \leq \tau(v') \leq x$. Since $v \in \lambda(\tau(v)) \cap \lambda(x)$, by (TD3), $v \in \lambda(\tau(w))$. By construction of the pebbling function, this implies $p(v) \neq p(w)$.

$(2) \Rightarrow (1)$. Suppose that $(V, {\leq}, p)$ is a $k$-pebble forest cover of $G$. We define a tree $T = V_{\bot}$ by adjoining a least element $\bot$  to $V$. We say that $v$ is an active predecessor of $v'$ if $v \leq v'$, and for all $w \in (v, v']$, $p(v) \neq p(w)$. We define the labelling function by setting $\lambda(v)$ to be the set of active predecessors of $v$; $\lambda(\bot) := \es$. Since $p |_{\lambda(v)}$ is injective, $|\lambda(v)| \leq k$.

We verify the tree decomposition conditions. (TD1) holds, since $v \in \lambda(v)$. (TD2) If $v \adj v'$, then $v \comp v'$. Suppose $v \leq v'$.
Then $v$ is an active predecessor of $v'$, and $\{ v, v' \} \subseteq \lambda(v')$.
(TD3) Suppose $v \in \lambda(v_1) \cap \lambda(v_2)$. Then $v$ is an active predecessor of both $v_1$ and $v_2$. This implies that for all $w \in \pth(v_1,v_2)$, $v$ is an active predecessor of $w$, and hence $v \in \lambda(w)$.
\end{proof}

\begin{theorem}
Let $\As$ be a finite $\sg$-structure. There is a bijective correspondence between:
\begin{enumerate}
\item $\Pk$-coalgebras $\alpha : \As \to \Pk \As$
\item $k$-pebble forest covers of $\Gf(\As)$.
\end{enumerate}
\end{theorem}
\begin{proof}
See \cite[Theorem 6]{abramsky2017pebbling}.
\end{proof}

\noindent We write $\cnP(\As)$ for the coalgebra number of $\As$ with respect to the the pebbling comonad.

\begin{theorem}
For all finite structures $\As$: $\tw(\As) \, = \, \cnP(\As) - 1$.
\end{theorem}

\subsection{The modal comonad and synchronization tree depth}.

Let $\As$ be a Kripke structure. It will be convenient to write labelled transitions $a \labar{\alpha} a'$ for $\Ralph(a, a')$.
Given $a \in A$, the submodel generated by $a$ is obtained by restricting the universe to the set of $a'$ such that there is a path $a \labar{\alpha_1} \cdots \labar{\alpha_k} a'$. This submodel forms a synchronization tree \cite{milner1980calculus} if for all $a'$, there is a unique such path. The height of such a tree is the maximum length of any path from the root $a$.

\begin{proposition}
Let $\As$ be a Kripke structure, with $a \in A$. The following are equivalent:
\begin{enumerate}
\item There is a coalgebra $\alpha : (\As, a) \to \Mk (\As, a)$.
\item The submodel generated by $a$ is a synchronization tree of height $\leq k$.
\end{enumerate}
\end{proposition}

\noindent We define the modal depth $\md(\As, a) = k$ if the submodel generated by $a$ is a synchronization tree of height $k$.

\begin{theorem}
Let $\As$ be a Kripke structure, and $a \in A$ be such that the submodel generated by $a$ is a synchronization tree of finite height.
Then $\md(\As, a) = \cnM(\As, a)$.
\end{theorem}

Note the conditional nature of this result, which contrasts with those for the other comonads. The modal comonad is defined in such a way that the universe $\Mk (A, a)$ reflects information about the possible transitions. Thus having a coalgebra at all, regardless of the value of the resource parameter, is a strong constraint on the structure of the transition system.

\section{Further Directions}

From the categorical perspective, there is considerable additional structure which we have not needed for the results in this paper, but which may be useful for further investigations.

\textbf{Coequaliser requirements}
In Moggi's work on computational monads, there is an ``equaliser requirement'' \cite{moggi1991notions}. The dual version for a comonad $(G, \varepsilon, \delta)$ is that for every object $A$, the following diagram is a coequaliser:
\begin{center}
\begin{tikzcd}
G^2 A \ar[r, "G \epsA", shift left=1.5ex] \ar[r, "\varepsilon_{GA}"']
& GA \ar[r, "\epsA"]
& A
\end{tikzcd}
\end{center}
This says in particular that the counit is a regular epi, and hence  $GA$ ``covers'' $A$ in a strong sense.

This coequaliser requirement holds for all our comonads. For $\Ek$, this is basically the observation that, given a sequence of sequences $[s_1, \ldots , s_j]$, we have $\ve[\ve s_1, \ldots , \ve s_j] \, = \, \ve s_j$. The other cases are similar.

\textbf{Indexed and graded structure}
Our comonads $\Ek$, $\Pk$, $\Mk$ are not merely discretely indexed by the resource parameter. In each case, there is a functor $(\Zp, {\leq}) \to \Comon(\CS)$
from the poset category of the positive integers to the category of comonads on $\CS$. Thus if $k \leq l$ there is a natural transformation with components $i^{k,l}_A : \Ek \As \to \El \As$, which preserves the counit and comultiplication; and similarly for the other comonads. Concretely, this is just including the plays of up to $k$ rounds in the plays of up to $l$ rounds, $k \leq l$.

Another way of parameterizing comonads by resource information is grading \cite{gaboardi2016combining}. Recall that comonads on $\CC$ are exactly the comonoids in the strict monoidal category $([\CC, \CC], {\circ}, I)$ of endofunctors on $\CC$ \cite{mac2013categories}. Generalizing this description, a graded comonad is an oplax monoidal functor $G : (M, {\cdot}, 1) \to ([\CC, \CC], {\circ}, I)$ from a monoid of grades into this endofunctor category. This means that for each $m \in M$, there is an endofunctor $G_m$, there is a graded counit natural transformation $\ve : G_{1} \Rightarrow I$, and for all $m, m' \in M$, there is a graded comultiplication $\delta^{m, m'} : G_{m\cdot m'} \Rightarrow G_m G_{m'}$.

The two notions can obviously be combined. We can see our comonads as (trivially) graded, by viewing them as oplax monoidal functors $(\Zp, {\leq}, \min, 1) \to ([\CC, \CC], {\circ}, I)$. Given $k \leq l$, we have e.g. $\Ek \Rightarrow \Ek \Ek \Rightarrow \Ek \El$.
The question is whether there are more interesting graded structures which arise naturally in considering richer logical and computational settings.

\textbf{Colimits and infinite behaviour}
In this paper, we have dealt exclusively with finite resource levels. However, there is an elegant means of passing to infinite levels. We shall illustrate this with the modal comonad. Using the inclusion morphisms described in the  previous discussion of indexed structure, for each structure $\As$ we have a diagram
\[ \MM_1 \As \to \MM_2 \As \to \cdots \to \Mk \As \to \cdots \]
By taking the colimits of these diagrams, we obtain a comonad $\Momega$, which corresponds to the usual unfolding of a Kripke structure to all finite levels. This will correspond to the bisimulation approximant $\sim_{\omega}$, which coincides with bisimulation itself on image-finite structures \cite{hennessy1980observing}. Transfinite extensions are also possible. Similar constructions can be applied to the other comonads. This provides a basis for lifting the comonadic analysis to the level of infinite models.

\textbf{Relations between fragments and parameters}
We can define morphisms between the different comonads we have discussed, which yield proofs about the relationships between the logical fragments they characterize. This categorical perspective avoids the cumbersome syntactic translations  in the standard proofs of these results. For illustration, there is a comonad morphism $t: \Ek \Rightarrow \Pk$ with components $t_A: \Ek \As \rarr \Pk \As$ given by $[a_{1},\dots,a_{j}] \mapsto [(1,a_{1}),\dots,(j,a_{j})]$. Together with  theorems \ref{thm:mainPebble} and \ref{thm:mainEF}, this shows that $\ELk \subseteq \ELvk$ and $\Lck \subseteq \Lvck$. Moreover, composing $t$ with a coalgebra $\As \rarr \Ek \As$ yields a coalgebra $A \rarr \Pk \As$, demonstrating that $\tw(\As) + 1 \leq \td(\As)$. Another morphism $\Momega \Rightarrow \PM_2$ shows that modal logic can be embedded into $2$-variable logic. 
   
\subsection*{Concluding remarks}
Our comonadic constructions for the three major forms of model comparison games show a striking unity, on the one hand, but also some very interesting differences. For the latter, we note the different forms of logical ``deception'' associated with each comonad, the different forms of back-and-forth equivalences, and the different combinatorial parameters which arise in each case.

One clear direction for future work is to gain a deeper understanding of what makes these constructions work. Another is to understand how widely the comonadic analysis of resources can be applied. We are currently investigating the guarded fragment \cite{andreka1998modal,gradel2014freedoms}; other natural candidates include existential second-order logic, and branching quantifiers and dependence logic \cite{vaananen2007dependence}.

Since comonads arise naturally in type theory and functional programming \cite{uustalu2008comonadic,orchard2014programming}, can we connect the study of finite model theory made here with a suitable type theory? Can this lead, via the Curry-Howard correspondence, to the systematic derivation of some significant meta-algorithms, such as decision procedures  for guarded logics based on the tree model property \cite{gradel1999decision}, or  algorithmic metatheorems such as Courcelle's theorem \cite{courcelle1990monadic}?

Another intriguing direction is to connect these ideas with the graded quantum monad studied in \cite{abramskyquantum}, which provides a basis for the study of quantum advantage in $\CS$. This may lead to a form of quantum finite model theory.

\bibstyle{plainurl}
\bibliography{bibfile}

\end{document}